\documentclass[1p]{elsarticle}

\usepackage{amsmath}
\usepackage{amssymb}
\usepackage{color}
\usepackage{latexsym}
\usepackage{soul}

\newtheorem{theorem}{Theorem}[section]
\newtheorem{lemma}[theorem]{Lemma}

\newtheorem{observation}[theorem]{Observation}

\newproof{proof}{Proof}   

\newcommand{\blue}{\textcolor{black}}

\newcommand{\steven}{\textcolor{black}}
\newcommand{\fudge}{\textcolor{black}}

\newcommand{\rSPR}{\mathrm{rSPR}}
\newcommand{\TBR}{{\rm TBR}}
\newcommand{\MP}{{\rm MP}}

\newcommand\delete[1]{}

\usepackage{amssymb}

\begin{document}

\begin{frontmatter}

\title{Cyclic generators and an improved linear kernel for the
rooted subtree prune and regraft distance}

\author[1]{Steven Kelk}
\author[2]{Simone Linz}
\author[1]{Ruben Meuwese \fnref{fn1}}

\fntext[fn1]{Ruben Meuwese was supported by the Dutch Research Council (NWO) KLEIN 1 grant \emph{Deep kernelization for phylogenetic discordance}, project number OCENW.KLEIN.305.}

\address[1]{Department of Data Science and Knowledge Engineering, Maastricht University, The Netherlands}
\address[2]{School of Computer Science, University of Auckland, New Zealand}

\begin{abstract}
The rooted subtree prune and regraft (rSPR) distance between two rooted binary phylogenetic trees is a well-studied measure of topological dissimilarity that is NP-hard to compute.
Here we describe an improved linear kernel for the problem. In particular, we show that if the
classical subtree and chain reduction rules are augmented with a modified type of chain reduction rule, the resulting trees have at most $9k-3$ leaves, where $k$ is the rSPR distance; and that this bound is tight. The previous best-known linear kernel had size $O(28k)$. To achieve this improvement  we introduce cyclic generators, which can be viewed as cyclic analogues of the generators used in the phylogenetic networks literature. \steven{As a corollary to our main result we also give an improved weighted linear kernel for the minimum hybridization problem on two rooted binary phylogenetic trees.}
\end{abstract}


\end{frontmatter}

\section{Introduction}
The central challenge of phylogenetics is to infer the evolutionary history of a set of contemporary species $X$. Often this history is modeled by a \emph{rooted phylogenetic tree}; essentially, a rooted tree in which the leaves are bijectively labeled by $X$ and evolution is explicitly directed away from the root \cite{steel2016phylogeny}. Due to confounding biological or methodological factors the inferred trees sometimes differ in topology, and then it is useful to formally quantify these differences \cite{HusonRuppScornavacca10}. One popular such difference measure is the \emph{rooted subtree prune and regraft} (rSPR) distance. Informally this measures the number of times that a subtree must be pruned, and re-attached, to transform one tree into another. Despite the NP-hardess of computing this distance \cite{bordewich2005computational}, very fast fixed-parameter tractable branching algorithms have been developed which allow the problem to be well solved in practice, as long as the rSPR distance does not become too large \cite{whidden2013fixed,
yamada2020improved}.
A related concept is \emph{kernelization}: polynomial-time pre-processing rules which reduce the size of the input trees to purely a function of their rSPR distance \cite{kernelization2019}. Compared to branching algorithms there \blue{has 
been} relatively little work on kernelization of rSPR. Indeed, currently the best-known result is that after exhaustive application of the \emph{subtree} and \emph{chain} reduction rules the
\delete{two}
input trees have at most $O(28k)$ leaves, where $k$ is the rSPR distance \cite{bordewich2005computational}.

\blue{In this paper,} we show that when a third, modified chain reduction rule is added to the portfolio, the bound improves to $9k-3$, and that this is in fact tight. To prove this we first show that computation of rSPR distance is essentially equivalent to the problem of parsimoniously embedding the two input trees into a potentially cyclic phylogenetic network (i.e. graph); it is a cyclic variant of the much-studied \steven{\emph{minimum hybridization}}
problem (see e.g. \cite{vanIersel20161075} and links therein). This allows us to introduce \emph{cyclic generators} which summarize the backbone of such networks, and allow us to carefully bound the size of reduced instances\delete{\footnote{Interestingly, the \steven{minimum hybridization} problem on two rooted trees was originally characterized as an acyclic variant of rSPR distance \cite{Semple:2007ug}, which in turn gave rise to acyclic generators \cite{approximationHN}. Here we are adapting these \emph{later} abstractions for the \emph{earlier} problem of rSPR distance.}}. Our approach is inspired by a similar strategy which has proven to be very powerful in the design of reduction rules for \emph{unrooted} phylogenetic trees \cite{kelk2020new}. \steven{As a corollary to our main rSPR result, we also show that the three aforementioned reduction rules yield a weighted linear kernel of $7k-2$ for the \steven{minimum hybridization} problem, where $k$ is the hybridization number of the two trees. This improves upon the weighted $9k-2$ kernel given in \cite{approximationHN}.}

\delete{We stress that all the \steven{rSPR} results in this article equally apply to kernelization of \emph{rooted maximum agreement forests} (on two rooted trees), due to the well-established equivalence between rSPR distance and such forests \cite{bordewich2005computational}.} 

\section{Preliminaries}

Throughout this paper, $X$ denotes a non-empty finite set. \\

\noindent {\bf Phylogenetic trees.} A {\em rooted phylogenetic $X$-tree} $T$ is a rooted tree with no degree-2 vertex, except for the root which has degree at least 2, and 
whose leaf set is $X$. All edges of $T$ are directed away from the root, i.e. if $(u,v)$ is an edge of $T$, then $u$ lies on the directed path from the root of $T$ to $v$. Furthermore, $T$ is {\it binary} if its root has degree 2 and all other interior vertices have degree 3. The leaf set $X$  is the {\it label set} of $T$ and denoted by $L(T)$. For two vertices $u$ and $v$ in $T$, we say that $u$ is an {\it ancestor} of $v$ if there is a directed path from the root of $T$ to $v$ that contains $u$.
We next define three types of subtrees of $T$ relative to \blue{a subset $X'\subseteq X$}. First, we write $T[X']$ to denote the minimal rooted subtree of $T$ that connects all elements in $X'$. Second, the {\it restriction} of $T$ to $X'$, denoted by $T|X'$, is the rooted phylogenetic $X'$-tree obtained from $T[X']$ by suppressing all vertices with in-degree 1 and out-degree 1. Lastly, a rooted subtree of $T$ is {\it pendant} if it can be detached from $T$ by deleting a single edge. Since all rooted phylogenetic trees throughout this paper are binary, we refer to a rooted binary phylogenetic tree simply as a {\it rooted phylogenetic tree}. For two rooted phylogenetic $X$-trees $T$ and $T'$, we say that $T$ and $T'$ are {\it isomorphic} if there is a bijection $\phi$ from the vertex set $V$ of $T$ to the vertex set of $T'$ such that $\phi(x)=x$ for each $x \in X$, and $(u,v)$ is an edge of $T$ if and only if $(\phi(u),\phi(v))$ is an edge of $T'$ for all $u,v\in V$. If $T$ and $T'$ are isomorphic, we write $T=T'$. \\

\noindent {\bf rSPR and agreement forests.} Let $T$ be a rooted phylogenetic $X$-tree. For the purposes of the upcoming definitions and indeed much of the paper, we view the root of $T$ as a vertex $\rho$ adjoined to the original root by a pendant edge. Furthermore, we regard $\rho$ as part of the label set of $T$, that is $L(T )=X\cup\{\rho\}$. Fig.~\ref{fig:trees} illustrates an example of two rooted phylogenetic $X$-trees with $X=\{x_1,x_2,\ldots,x_6\}$ with their roots labeled with $\rho$. Let $e=(u,v)$ be an edge of $T$ not incident with $\rho$. Let $T'$ be the rooted phylogenetic $X$-tree obtained from $T$ by deleting $e$ and re-attaching the resulting rooted subtree containing $v$ via a new edge$f$ as follows. Subdivide an edge of the component that contains $\rho$ with a new vertex $u'$, join $u'$ and $v$ with $f$, and suppress $u$. We say that $T'$ has been obtained from $T$ by a {\it rooted subtree prune and regraft} (rSPR) operation.
The {\it $\rSPR$ distance}  between any two rooted phylogenetic $X$-trees $T$ and $T'$, denoted by $d_\rSPR(T,T')$, is the minimum number of $\rSPR$ operations that transform $T$ into $T'$. It is well known that one can always transform $T$ into $T'$ via a sequence of $\rSPR$ operations. However, computing $d_\rSPR(T,T')$ is an NP-hard problem~\cite{bordewich2005computational,hein1996complexity}.

\begin{figure}
\center
\scalebox{1}{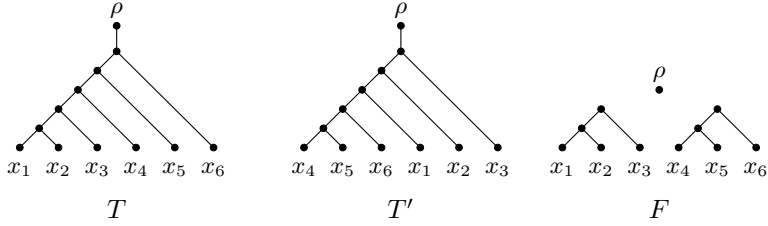}
\caption{Two rooted phylogenetic trees $T$ and $T'$ with their roots labeled $\rho$, and an agreement forest $F$ for $T$ and $T'$. All edges are directed downwards.} 
\label{fig:trees}
\end{figure}

Now, let $T$ and $T'$ be two rooted phylogenetic $X$-trees.  An {\it agreement forest} $F=\{L_\rho,L_1,\ldots,L_k\}$ for $T$ and $T'$ is a partition of $X\cup\{\rho\}$ such that $\rho\in L_\rho$ and the following two properties are satisfied:
\begin{enumerate}[(P1)]
\item  For all $i\in \{\rho,1,\ldots,k\}$, we have $T|L_i=T'|L_i$.
\item The trees in $\{T[L_i]: i\in \{\rho,1,\ldots,k\}\}$ and 
$\{T'[L_i]: i\in \{\rho,1,\ldots,k\}\}$ are vertex-disjoint subtrees of $T$ and $T'$, respectively. 
\end{enumerate}
An agreement forest for $T$ and $T'$ is a {\em maximum agreement forest} if, amongst all agreement forests for $T$ and $T'$, it has the smallest number of elements. To illustrate, Fig.~\ref{fig:trees} shows an agreement forest $F$ for the two rooted phylogenetic trees $T$ and $T'$ of the same figure. Indeed, $F$ is a maximum agreement forest for $T$ and $T'$. The following theorem characterizes the rSPR distance between two rooted phylogenetic trees (with their roots labeled $\rho$) in terms of agreement forests.

\begin{theorem}\cite{bordewich2005computational}\label{t:first-characterization}
Let $T$ and $T'$ be two rooted phylogenetic $X$-trees, and let $F$ be a maximum agreement forest for $T$ and $T'$. Then $d_\rSPR(T,T')=|F|-1$.
\end{theorem}

\section{\blue{Leaf-labeled} graphs characterize the rSPR distance}
In this section, we establish an alternative characterization for the rSPR distance between two rooted phylogenetic trees.  A {\em rooted leaf-labeled graph $G$ on $X$} is a rooted directed 
graph with no parallel edges \steven{or loops} that satisfies the following four properties:
\begin{enumerate}[(i)]
\item \steven{the unique root has in-degree 0 and out-degree 1, and is labeled $\rho$,}
\item a vertex of out-degree 0 has in-degree 1, and the set of vertices with out-degree 0 is $X$, 
\item all other vertices either have in-degree 1 and out-degree 2, or in-degree 2 and out-degree 1, and
\item each vertex can be reached from \blue{$\rho$} via a directed path.
\end{enumerate}
A vertex of $G$ with in-degree 1 and out-degree 2 is a {\em tree vertex}, while a vertex of in-degree 2 and out-degree 1 is a {\em reticulation}. For two vertices $u$ and $v$ in $G$, we say that $u$ is a {\it parent} of $v$ if $(u,v)$ is an edge.  In contrast to a rooted phylogenetic network~\cite{HusonRuppScornavacca10,steel2016phylogeny}, observe that $G$ may contain a directed cycle. \steven{Nevertheless, as in the case of rooted (binary) phylogenetic networks, the number of reticulations in $G=(V,E)$, denoted $r(G)$, is  
\blue{equal} to $|E|-(|V|-1)$. This is because $G$, due to property (iv), has a directed spanning tree, rooted
at $\rho$, with $|V|-1$ edges. The spanning tree does not yet have any vertices with in-degree 2. Each of the $|E|-(|V|-1)$ edges from $E$ that are not on the spanning tree, creates exactly one in-degree 2 vertex when added to it. Hence, there are exactly  $|E|-(|V|-1)$ reticulations in total. \delete{(Note that property (iv) also entails that any directed cycle on $G$ must contain at least one reticulation).}}

As for rooted phylogenetic trees, a rooted subtree of $G$ is {\it pendant} if it can be detached from $G$ by deleting a single edge. Let $T$ and $T'$ be two rooted phylogenetic $X$-trees. We say that $T$ is {\it displayed} by $G$ if  there exists a subgraph of $G$ that is a subdivision of $T$. Moreover we set 
$$r^\circ(T,T')=\min_G\{r(G)\}.$$  \steven{That is,}
 $r^\circ(T,T')$ equates to the minimum number of reticulations over all rooted leaf-labeled graphs that display $T$ and $T'$. Fig.~\ref{fig:graph-generator} shows a rooted leaf-labeled graph $G$ that displays the two rooted phylogenetic trees $T$ and $T'$ that are depicted in Fig.~\ref{fig:trees}. Note that $r(G)=2=d_\rSPR(T,T')$. The next theorem shows that this relationship is not a coincidence. We note that the idea of viewing a sequence of rSPR operations as a rooted leaf-labeled graph was briefly mentioned in~\cite{Semple:2007ug} for the purpose of highlighting that such a graph may contain a directed cycle. 

\begin{figure}
\center
\scalebox{1}{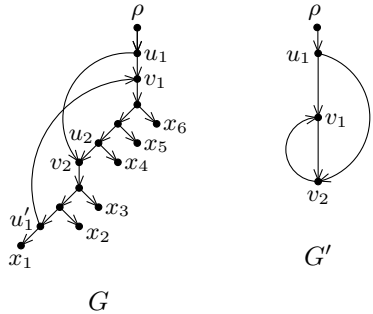}
\caption{Left: A rooted leaf-labeled graph $G$ that displays the two rooted phylogenetic trees $T$ and $T'$ that are shown in Fig.~\ref{fig:trees}. The reticulations of $G$ are $v_1$ and $v_2$. To see that $G$ displays $T'$, note that the graph obtained from $G$ by deleting the two edges $(u_1,v_1)$ and $(u_2,v_2)$ is a subdivision of $T'$. Right: The cyclic 2-generator $G'$ that underlies $G$. To obtain $G$ from $G'$, the elements in $\{x_1,x_2,x_3\}$ are attached to the side $(v_2,v_1)$ and the elements in $\{x_4,x_5,x_6\}$ are attached to the side $(v_1,v_2)$.}
\label{fig:graph-generator}
\end{figure}

\begin{theorem}\label{t:second-characterization}
Let $T$ and $T'$ be two rooted phylogenetic $X$-trees. Then $d_\rSPR(T,T')=r^\circ(T,T')$.
\end{theorem}

\begin{proof}
Throughout this proof, we \steven{continue with our convention that trees and graphs have an in-degree 0, out-degree 1 root labeled $\rho$.}

We first show that $r^\circ(T,T')\geq d_\rSPR(T,T')$. This part of the proof is similar to the second part of the proof of~\cite{van2018unrooted}. Let $G$ be a rooted leaf-labeled graph on $X$ that displays $T$ and $T'$ such that $r(G)=r^\circ(T,T')$. Let $V(G)$ and $E(G)$ be the vertex and edge set of $G$, respectively. Let $E_T$ be the edge set of a subdivision of $T$ in $G$. Similarly, let $E_{T'}$ be the edge set of a (directed) spanning tree of $G$ that is obtained from a subdivision of $T'$ in $G$ by adding a possibly empty set of edges. Note that $|E_{T'}|=|V(G)|-1$, and that both $E_T$ and $E_{T'}$ contain the edge of $G$ that is incident with $\rho$. Lastly, let $A$ be the subset of $E_T$ that contains precisely each edge that is not in $E_{T'}$. We next obtain two graphs from $G$. First, obtain $G'$ from $G$ by deleting each edge in $A$. Observe that the edge set of $G'$ contains each edge in $E_{T'}$ and, hence $|E(G)|-|A|\geq |V(G)|-1$. It therefore follows that 
\begin{equation}\label{eq:one}
|A|\leq |E(G)|-|V(G)|+1=r(G).
\end{equation}
Second, obtain $F$ from $G$ by deleting each edge that is not in $E_T$, deleting each edge in $A$, deleting each of the resulting connected components that does not contain at least one vertex labeled with an element in $X\cup\{\rho\}$, and applying any of the following operations until no further operation is possible.
\begin{enumerate}
\item Delete each vertex with in-degree 0 and out-degree 1 that is not $\rho$.
\item Delete each unlabeled vertex with out-degree 0.
\item Suppress each vertex with in-degree 1 and out-degree 1.
\end{enumerate}
By construction, $F$ has  at most $|A|+1$ elements. Furthermore, the partition of $X\cup\{\rho\}$ in which each block corresponds to the label set of an element in $F$ is an agreement forest for $T$ and $T'$. Hence,$$d_\rSPR(T,T')\leq |F|-1\leq |A|\leq r(G)=r^\circ(T,T'),$$ where the first inequality follows from Theorem~\ref{t:first-characterization} and the third inequality follows from Equation~\ref{eq:one}.

We complete the proof by showing that $r^\circ(T,T')\leq d_\rSPR(T,T')$. This part of the proof is by induction on $d_\rSPR(T,T')$. If $d_\rSPR(T,T')=0$, then $G=T=T'$ is a rooted leaf-labeled graph with $r(G)=0$ that displays $T$ and $T'$. Assume that $d_\rSPR(T,T')=k$ and that the theorem holds for all pairs of rooted phylogenetic trees whose rSPR distance is at most $k-1$.  Then there exists a rooted phylogenetic $X$-tree $T''$ such that $d_\rSPR(T,T'')=k-1$ and $d_\rSPR(T'',T')=1$. \blue{(If $k=1$, then $T=T''$.)} By the induction assumption, there exists a rooted leaf-labeled graph $G'$ on $X$ with $r(G') \leq k-1$ that displays $T$ \mbox{and~$T''$.} 

We next construct a rooted leaf-labeled graph $G$ from $G'$. Let $E_{T''}$ be the edge set of a subdivision of $T''$ in $G'$. Consider the rSPR operation that transforms $T''$ into $T'$. Let $f$ be the edge that is deleted in $T''$ and let $f'$ be the edge that is subdivided after the deletion of $f$. Then $f$ (resp. $f'$) corresponds to a directed path $P$ (resp. $P'$) in $E_{T''}$. Let $e$ (resp. $e'$) be an edge of $P$ (resp. $P'$).
Now obtain $G$ from $G'$ by subdividing $e$ with a new vertex $v_1$, subdividing $e'$ with a new vertex $v_2$, and adding the edge $(v_2,v_1)$. Clearly as $G'$ is a rooted leaf-labeled graph on $X$, $G$ is also such a graph with $r(G)=r(G')+1$. Moreover, as $G'$ displays $T$ and $T''$, it follows from the construction that $G$ displays $T$ and $T'$. Hence 
\begin{align*}
d_\rSPR(T,T') &= d_\rSPR(T,T'')+d_\rSPR(T'',T')  \geq  r(G')+1=r(G)\geq r^\circ(T,T').
\end{align*}\qed
\end{proof}

\section{Cyclic generators}

Let $k\geq 1$ 
be a positive integer. A {\it cyclic $k$-generator} (or short {\it cyclic generator} if $k$ is clear from the context) is a connected directed graph that may contain parallel edges but no loops, and that satisfies the following four properties:
\begin{enumerate}[(i)]
\item the unique root is labeled $\rho$ and has in-degree 0 and  out-degree 1, 
\item  there are exactly~$k$ vertices with in-degree 2 and out-degree at most 1, 
\item all other vertices have in-degree 1 and out-degree 2, and
\item each vertex can be reached from $\rho$ via a directed path.  
 \end{enumerate}
 \noindent The {\it sides} of a cyclic $k$-generator are its edges, called the {\it edge sides}, and its vertices of in-degree 2 and out-degree 0, called the {\it vertex sides}.

Now, let $G$ be a rooted leaf-labeled graph with $r(G)=k$ that has no pendant subtree with at least two leaves. 
Then, we can obtain a cyclic $k$-generator $G'$ from $G$ by deleting all leaves and suppressing each resulting vertex with in-degree 1 and out-degree 1.  We say that $G'$ is the cyclic $k$-generator that {\it underlies} $G$. Reversely, the edge and vertex sides of a cyclic generator are the places where leaves can be attached to obtain a rooted leaf-labeled graph. More precisely, let $Y=\{y_1,y_2,\ldots,y_m\}$ be a set of leaves, and let $G'$ be a cyclic $k$-generator. Then, {\it attaching} $Y$ to an edge side $(u,v)$ of $G'$ is the operation of subdividing $(u, v)$ with $m$  vertices $w_1,w_2,\ldots,w_m$ and, for each $i\in\{1,2,\ldots,m\}$, adding an edge $(w_i, y_i)$. Moreover, {\it attaching} $Y$ to a vertex side $v$ of $G'$ is the operation of adding an edge $(v, r)$, where $r$ is the root of a rooted phylogenetic $Y$-tree. If at least one new leaf is attached to each pair of parallel edges and to each vertex side in $G'$, then the resulting graph is a rooted leaf-labeled graph $G$ with $r(G) = k$. We summarize the construction in the next observation.

\begin{observation}
Let $G$ be a rooted leaf-labeled graph that has no pendant subtree with at least two leaves, and let $G'$ be a cyclic $r(G)$-generator. Then $G'$ underlies $G$ if and only if $G$ can be obtained from $G'$ by attaching a (possibly empty) set of leaves to each edge and vertex side of $G$.
\end{observation}

\noindent As an example, Fig.~\ref{fig:graph-generator} shows the cyclic $2$-generator $G'$ that underlies the rooted leaf-labeled graph $G$ that is depicted in the same figure.

The proof of the next lemma was first established in~\cite{approximationHN}.

\begin{lemma}\label{l:generator-edges}
Let $k\geq 1$, and let $G'$ be a cyclic $k$-generator. Then $G'$ has $4k_0+3k_1-1$ edge sides, where $k_0$ is the number of vertex sides in $G'$ and $k_1$ is the number of vertices in $G'$ with in-degree 2 and out-degree 1.
\end{lemma}

\section{Reductions}

This section describes three reductions that can be applied to two rooted  phylogenetic trees to shrink them to two smaller trees before computing their rSPR distance. The first two reductions were established in~\cite{bordewich2005computational}, where the authors have  shown that each reduction preserves the rSPR distance. The third reduction, which was established in~\cite{whidden2013fixed} in the context of a depth-bounded search tree algorithm for computing the rSPR distance reduces the rSPR distance by 1.

Let $T$ be a rooted phylogenetic $X$-tree, and let $C=(x_1,x_2,\ldots,x_n)$ be a sequence of elements in \steven{$X$}
with $n\geq 2$.  We say that $C$ is an {\it $n$-chain}  (or short {\it chain}) of $T$ if the parent of $x_1$ coincides with the parent of $x_2$ or the parent of $x_2$ is the parent of the parent of $x_1$, and, for each $i\in\{3,4,\ldots,n\}$, the parent of $x_i$ is the parent of \blue{the parent of} $x_{i-1}$. By definition, no chain of $T$ contains $\rho$. If $C$ is a chain of $T$ and the parent of $x_1$ coincides with the parent of $x_2$, then we say that $C$ is {\it pendant} in $T$, \blue{in which case $C=(x_1,x_2,x_3,\ldots,x_n)=(x_2,x_1,x_3,\ldots,x_n)$}. \steven{If a \blue{chain}  
is a chain of both $T$ and $T'$, we say that it is a \emph{common} chain.} Referring back to Fig.~\ref{fig:trees}, we note that $T$ and $T'$ as shown in this figure have two common $3$-chains $(x_1,x_2,x_3)$ and $(x_4,x_5,x_6)$ and each is pendant in one of $T$ and $T'$.

Let $T$ and $T'$ be two rooted phylogenetic $X$-trees. 
\blue{We next} describe three reductions to obtain two rooted phylogenetic trees $S$ and $S'$ from $T$ and $T'$, respectively, with fewer leaves.\\

\noindent{\bf Subtree reduction.} For $m\geq 2$, let $\{x_1,x_2,\ldots,x_m\}$ be the leaf set of a maximal pendant subtree that is common to $T$ and $T'$. Then set $S=T|X\setminus\{x_2,x_3,\ldots,x_m\}$ and $S'=T'|X\setminus\{x_2,x_3,\ldots,x_m\}$.\\

\noindent{\bf Chain reduction.} For $n\geq 4$, let $C=(x_1,x_2,\ldots,x_n)$ be a maximal $n$-chain that is common to $T$ and $T'$. Then set $S=T|X\setminus\{x_4,x_5,\ldots,x_n\}$ and $S'=T'|X\setminus\{x_4,x_5,\ldots,x_n\}$.\\

\noindent{\bf 3-2-chain reduction.} Let $(x_1,x_2,x_3)$ be a pendant 3-chain of $T$. If $(x_i,x_3)$ is a pendant 2-chain in $T'$ with $x_i\in\{x_1,x_2\}$, then set $S=T|X\setminus\{x_j\}$ and $S'=T'|X\setminus\{x_j\}$ with $\{x_i,x_j\}=\{x_1,x_2\}$. \\

Note that \steven{after an application of the 3-2-chain reduction, $(x_i,x_3)$ is a pendant 2-chain that is common to $S$ and $S'$. It can therefore be further reduced by a subtree reduction.}

The next lemma shows that an application of the 3-2-chain reduction reduces the rSPR distance by 1. A slightly more general result was established in~\cite{whidden2013fixed}, where the authors applied the reduction to two forests instead of to two rooted phylogenetic trees. To keep the exposition self contained, we  include a full proof that is adapted to the setting of our paper.

\begin{lemma}\label{l:3-2-reduction}
Let $T$ and $T'$ be two rooted phylogenetic $X$-trees, and let $S$ and $S'$ be two trees obtained from $T$ and $T'$, respectively, by a single application of the 3-2-chain reduction. Then $d_\rSPR(S,S')=d_\rSPR(T,T')-1$.
\end{lemma}

\begin{proof}
Without loss of generality, we establish the lemma using the same notation as in the definition of a 3-2-chain reduction. Let $F_S$ be a maximum agreement forest for $S$ and $S'$, and let $F_T$ be a maximum agreement forest for $T$ and $T'$. Then $F_S\cup\{\{x_j\}\}$ is an agreement forest for $T$ and $T'$, which implies that $|F_S|+1\geq |F_T|$. Hence $$d_\rSPR(S,S')=|F_S|-1\geq |F_T|-2=d_\rSPR(T,T')-1.$$

Now consider $F_T$. \steven{If $\{x_j\} \in F_T$ then $F \setminus \{ \{x_j\}\}$ is an agreement forest for $S$ and $S'$, so $d_\rSPR(S,S') \leq d_\rSPR(T,T')-1$ and we are done.} Assume \steven{therefore} that $\{x_j\}\notin F_T$. Let $B$ be the element in $F_T$, with $|B|\geq 2$, that properly contains $x_j$. Then (P2) in the definition of an agreement forest implies that $x_i$ and $x_3$ cannot both be contained in $B$. We next consider three cases.

First, assume that $x_3\in B$ and $x_i\notin B$. Then $\{x_i\}\in F_T$. Let $B'=(B\setminus\{x_j\})\cup \{x_i\}$. Since $T|B=T'|B$, it follows that $$(F_T\setminus\{B,\{x_i\}\})\cup\{\{x_j\},B'\}$$ is a maximum agreement forest for $T$ and $T'$. Second, assume that $x_i\in B$ and $x_3\notin B$. Then $\{x_3\}\in F_T$ and an argument that is similar to that used in the first case implies that there exists a maximum agreement forest for $T$ and $T'$ in which $\{x_j\}$ is an element. Third, assume that $x_i,x_3\notin B$. Then, as $F_T$ satisfies (P2), $\{x_i\}$ and $\{x_3\}$ are both elements in $F_T$. Hence $$(F_T\setminus \{B,\{x_i\},\{x_3\}\})\cup \{\{x_i,x_3\},\{x_j\},B\setminus \{x_j\}\}$$ is a maximum agreement forest for $T$ and $T'$.

Taken together, the three cases described in the last paragraph show that there exists \steven{another} maximum agreement forest for $T$ and $T'$ in which $\{x_j\}$ is an element. We may therefore assume that $F_T$ is indeed such a forest. This implies that $F_T\setminus \{\{x_j\}\}$ is an agreement forest for $S$ and $S'$ with $|F_S|\leq |F_T|-1$ and, so, $$d_\rSPR(T,T')-1=|F_T|-2\geq |F_S|-1=d_\rSPR(S,S').$$ Combining both cases establishes the lemma.\qed
\end{proof}

\section{\steven{A new kernel for rSPR distance}}
The current smallest kernel size for computing the rSPR distance as stated in the next lemma was established in 2005~\cite{bordewich2005computational}.

\begin{lemma}
Let $S$ and $S'$ be two rooted phylogenetic $X$-trees. Suppose that $S$ and $S'$ cannot be reduced any further by applying the subtree or chain reduction. 
Then $|X|\leq 28d_\rSPR(S,S')$.
\end{lemma}

We next show that the size of the rSPR kernel can be substantially improved by additionally applying the 3-2-chain reduction.
\begin{theorem}\label{t:kernel}
Let $S$ and $S'$ be two rooted phylogenetic $X$-trees such that $d_\rSPR(S,S')\geq 1$. 
Suppose that $S$ and $S'$ cannot be reduced any further by applying the subtree, chain, or 3-2-chain reduction. Then $|X|\leq 9d_\rSPR(S,S')-3$.
\end{theorem}

\begin{proof}
Let $G$ be a rooted leaf-labeled graph on $X$ that displays $S$ and $S'$ such that $r(G)=r^\circ(S,S')=d_\rSPR(S,S')=k\geq 1$, where the second equality follows from Theorem~\ref{t:second-characterization}. Let $G'$ be the cyclic $k$-generator that underlies $G$. Now $G$ can be obtained \steven{from} $G'$ by attaching leaves in $X$ to the edge and vertex sides of $G'$. In what follows we bound the number of leaves that can be attached to three different types of such sides in $G'$. First, let $v$ be a vertex with in-degree 2 and out-degree 0. If no leaf is attached in obtaining $G$ from $G'$, then $G$ is not a rooted leaf-labeled graph. Moreover, if at least two leaves are attached to $G'$, then $S$ and $S'$ have a common pendant subtree with at least two leaves and can be further reduced by applying the subtree reduction. Hence, $G$ is obtained from $G'$ by attaching exactly one leaf to $v$. Second, let $e=(u,v)$ and $e'=(u',v)$ be two edge sides such that $v$ is a vertex side. Let $x_1$ be the unique leaf that is attached to $v$ in obtaining $G$ from $G'$. Now assume that at least two leaves $x_2$ and $x_3$ are attached to one of $e$ and $e'$, say \blue{$e$}.
Without loss of generality, we may assume that $(p_3,p_2)$ and $(p_2,v)$ are edges in $G$, where $p_2$ and $p_3$ are the parent of $x_2$ and $x_3$, respectively. Since $r(G)=r^\circ(S,S')$, it follows that, regardless of how many leaves are attached to $e'$,  $(x_1,x_2,x_3)$ is a pendant 3-chain in one of $S$ and $S'$, and $(x_2,x_3)$
is a pendant 2-chain in the other tree. \steven{This is because, if we consider subdivisions of $S$ and $S'$ in $G$, at least one of the two subdivisions does \emph{not} use the edge $(p_2,v)$. If both used edge $(p_2, v)$, then the other edge entering $v$ would not be used by either subdivision, and could safely be deleted, contradicting the assumed minimality of $G$ i.e.   $r(G)=r^\circ(S,S')$.} Consequently, $S$ and $S'$ can be further reduced by applying the 3-2-chain reduction. Hence $G$ is obtained from $G'$ by attaching at most one leaf to $e$ and at most one leaf to $e'$. Third, let $e$ be an edge side that is not directed into a vertex side. If at least four leaves are attached in obtaining $G$ from $G'$, then $S$ and $S'$ have a common 4-chain and can be further reduced by the chain reduction. Hence $G$ is obtained from $G'$ by attaching at most three leaves to $e$. Now, in $G'$, let $k_0$ be the number of vertex sides, and let $k_1$ be the number of vertices with in-degree 2 and out-degree 1. Then $k=k_0+k_1$. Moreover, by Lemma~\ref{l:generator-edges}, $G'$ has $4k_0+3k_1-1=2k_0+2k_0+3k_1-1$ edge sides. Since there are $2k_0$ edge sides that are directed into a vertex side and $2k_0+3k_1-1$ edge sides that are not directed into a vertex side, we have
\begin{align*}
|X| & \leq  1\cdot 2k_0+3(2k_0+3k_1-1)+1\cdot k_0=9k_0+9k_1-3=9k-3=9d_\rSPR(S,S')-3.
\end{align*}\qed
\end{proof}

We next establish that the bound as stated in the last theorem is tight. \delete{That is, we show that there exist pairs of rooted phylogenetic trees $S$ and $S'$ of arbitrary size that cannot be reduced under the subtree, chain, or 3-2-chain reduction, and whose number of leaves is \steven{$9d_\rSPR(S,S')-3$}.} The approach we take is similar to that of~\cite[Theorem 6]{tightkernel}. We start by briefly introducing some new definitions and refer the interested reader to~\cite{tightkernel} (and references therein such as~\cite{AllenSteel2001,fischer2014,moulton2015}) for full details. A {\it binary character} $f$ on $X$ is a function that assigns each element in $X$ to an element in $\{0,1\}$. Let $T$ be an unrooted binary phylogenetic $X$-tree with vertex set $V$, that is, $T$ can be obtained from a rooted binary phylogenetic $X$-tree \steven{(without $\rho$)} by suppressing its root with in-degree 0 and out-degree 2. An {\it extension} $g$ of $f$ to $V$ is  a function $g$ that assigns each element in $V$ to an element in $\{0,1\}$ such that $g(x)=f(x)$ for each $x\in X$. The {\it parsimony score} of $f$ on $T$, denoted by $l_f(T)$, denotes the minimum number of edges $\{u,v\}$ in $T$ such that $g(u)\ne g(v)$, \steven{ranging} over all extensions of $f$. Now, for two unrooted binary phylogenetic $X$-trees $T$ and $T'$, the {\it maximum parsimony distance \steven{on binary characters}} \steven{$d^2_{\MP}$} is defined as  $$\steven{d^2_{\MP}}(T,T')=\max_f |l_f(T)-l_f(T')|,$$ \steven{where $f$ ranges over all binary characters on $X$}. Lastly, the {\it tree bisection and reconnection distance} $d_\TBR(T,T')$ between $T$ and $T'$ can, informally, be viewed as the minimum number of operations needed to transform $T$ \steven{into} $T'$, where each operation consists of deleting an edge in a tree and then re-attaching the two resulting (smaller) trees back together by joining them with a new edge. \steven{It is an unrooted analogue of the rSPR distance.} \steven{Indeed, similar} to Theorem~\ref{t:first-characterization}, Allen and Steel~\cite{AllenSteel2001} have shown that $d_\TBR(T,T')$ can be characterized by (unrooted) maximum agreement forests. In what follows, the maximum parsimony distance \steven{on binary characters} and the TBR distance between $T$ and $T'$ will play an important role because \steven{$d^2_\MP(T,T')$} is a lower bound on $d_\TBR(T,T')$~\cite{fischer2014} and $d_\TBR(T,T')$ is a lower bound on the rSPR distance between two rooted binary phylogenetic $X$-trees that can be obtained by rooting $T$ and $T'$. We now make this more precise.

\begin{figure}
\center
\scalebox{1}{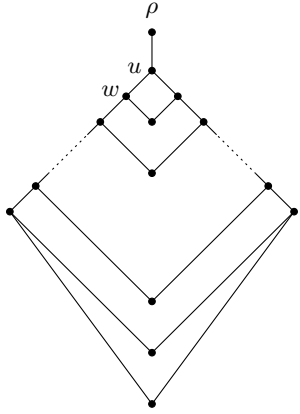}
\caption{For $k\geq 1$, the cyclic $k$-generator $G_k'$ used in the construction of a family of pairs of rooted phylogenetic trees to show that the linear kernel established in Theorem~\ref{t:kernel} is tight. All edges are directed downwards.} 
\label{fig:tight}
\end{figure}

\begin{theorem}\label{t:tight}
Let $S$ and $S'$ be two rooted phylogenetic $X$-trees such that $d_\rSPR(S,S')\geq 1$. 
Suppose that $S$ and $S'$ cannot be reduced any further by applying the subtree, chain, or 3-2-chain reduction. Then $|X|\leq 9d_\rSPR(S,S')-3$ is a tight bound.
\end{theorem}

\begin{proof}
Let $k\geq 1$, and let $G_k'$ be the cyclic $k$-generator that is shown in Fig.~\ref{fig:tight}. Observe that $G'_k$ has $k$ vertex sides, $2k$ edge sides that are directed into a vertex side and $2k-1$ edge sides that are not directed into a vertex side. Obtain a rooted leaf-labeled graph $G_k$ on $X\cup\{\rho\}$ from $G_k'$ by attaching one leaf to each vertex side and to each edge side directed into a vertex side, and attaching three leaves to each remaining edge side. Then $|X|=k+2k+3(2k-1)=9k-3$. In what follows, we say that an edge $(u,v)$ that is directed into a reticulation in $G_k$ is {\it a left reticulation edge} (resp. {\it right reticulation edge}) if, in the process of obtaining $G_k$ from $G_k'$, $u$ subdivides an edge side $(p,v)$ of $G_k'$, where $v$ is a vertex side and \steven{$p$} is to the \steven{left} (resp. \steven{right}) of $v$ in Fig.~\ref{fig:tight}.  Now, let $S_k$ be the rooted phylogenetic tree with label set $X\cup\{\rho\}$ obtained from $G_k$ by deleting all right reticulation edges and suppressing all resulting vertices of in-degree 1 and out-degree 1. Similarly, let $S_k'$ be the  rooted phylogenetic $X$-tree obtained from $G_k$ by deleting all left reticulation edges and suppressing all resulting vertices of in-degree 1 and out-degree 1. It is straightforward to check 
that $S_k$ and $S_k'$ cannot be reduced under the subtree, chain, or 3-2-chain reduction.  We next show that $d_\rSPR(S_k,S_k')=k$. By construction, $S_k$ and $S_k'$ are displayed by $G_k$ and, so $d_\rSPR(S_k,S_k')\leq k$. It remains to show that $d_\rSPR(S_k,S_k')\geq k$.
\steven{The claim holds immediately when $k=1$ because $S_k \neq S_k'$. Hence, we assume that $k \geq 2$.}
Let $\bar{S}_k$ and $\bar{S}_k'$ be the two unrooted binary phylogenetic $X$-trees obtained from $S_k$ and $S_k'$, respectively, by deleting $\rho$, suppressing the resulting vertex of in-degree 0 and out-degree 2, and ignoring the directions on the edges. Consider the edge side $(u,w)$ of $G_k'$ as shown in Fig.~\ref{fig:tight}. By construction, \steven{and because $k \geq 2$,}
there is a directed path $(u,v_1),(v_1,v_2),(v_2,v_3),(v_3,w)$ in $G_k$ and, therefore, also in $S_k$. Now let $f$ be the binary character that assigns 0 to each element in $X$ if and only if it is a descendant of $v_1$ in $S_k$. Then $l_f(\bar{S}_k)=1$. On the other hand, by applying the well-known Fitch algorithm~\cite{fitch1971}, we see that $l_f(\bar{S}_k')=k+1$ and, thus, $$|1-(k+1)|=k\leq \steven{d^2_\MP}(\bar{S}_k,\bar{S}_k')\leq d_\TBR(\bar{S}_k,\bar{S}_k'),$$ where the last inequality is established in~\cite{fischer2014}. We next show~that $d_\TBR(\bar{S}_k,\bar{S}_k')$ is a lower bound on $d_\rSPR(S_k,S_k')$. Let $F_k$ be a maximum agreement forest for $S_k$ and $S_k'$. Let $L_\rho$ be the element in $F_k$ such that $\rho\in L_\rho$. Then, the forest $\bar{F}_k$  obtained from $F_k$ by replacing $L_\rho$ with $L_\rho\setminus \{\rho\}$ is an (unrooted) agreement forest for $\bar{S}_k$ and $\bar{S}_k'$ with $|\bar{F}_k|\leq |F_k|$. In summary, we have 
\begin{align*}
k& \leq \steven{d^2_\MP}(\bar{S}_k,\bar{S}_k')\leq d_\TBR(\bar{S}_k,\bar{S}_k')\leq |\bar{F}_k|-1\leq |F_k|-1=d_\rSPR(S_k,S_k'),
\end{align*}
where the third inequality follows from~\cite{AllenSteel2001}. Setting $S=S_k$ and $S'=S_k'$, the theorem now follows.\qed
\end{proof}

\section{Minimum Hybridization}
In this section, we turn to rooted leaf-labeled graphs without any directed cycle which are known as {\it rooted phylogenetic networks}. \delete{Phylogenetic networks are becoming increasingly important in the study of evolutionary processes such as lateral gene transfer and hybridization that cannot be represented by a single rooted phylogenetic tree.} In this context, computing the {\it hybridization number} 
$$r(T,T')=\min_N\{r(N)\},$$
\blue{where the minimum is taken over all rooted phylogenetic networks that display $T$ and $T'$,} has attracted much interest over the last 15 years. 
The hybridization number can \fudge{also} be characterized in terms of agreement forests. Let $F=\{L_\rho,L_1,\ldots,L_k\}$ be an agreement forest for $T$ and $T'$. Then $F$ is {\it acyclic} if the graph $G_F$ with vertex set $F$ and for which $(L_i,L_j)$ with $i,j\in\{\rho,1,\ldots,k\}$ is an edge precisely if 

\begin{enumerate}[(i)]
\item the root of $T[L_i]$ is an ancestor of the root of $T[L_j]$, or
\item the root of $T'[L_i]$ is an ancestor of the root of $T'[L_j]$.
\end{enumerate}
does not contain a directed cycle. Moreover, a {\it maximum acyclic agreement forest} for $T$ and $T'$ is an acyclic agreement forest for $T$ and $T'$ whose number of elements is minimum.

\begin{theorem}\cite{baroni2005bounding}
Let $T$ and $T'$ be two rooted phylogenetic $X$-trees, and let $F$ be a maximum acyclic agreement forest for $T$ and $T'$. Then $r(T,T')=|F|-1$.
\end{theorem}

\noindent
Computing $r(T,T')$ is known to be NP-hard but fixed-parameter tractable~\cite{sempbordfpt2007,bordewich2007computing}, and the current best \steven{weighted} kernel\delete{\footnote{\steven{Viewed without weighting, the hybridization number kernel is quadratic in size \cite{ierselLinz2013}.}}} has size $O(9k)$, where $k=r(T,T')$~\cite{approximationHN}. This
result relies on applying the subtree reduction and the following modified chain reduction that reduces a common $n$-chain
to a \steven{(weighted)} $2$-chain, and
{\it $k$-generators} which are cyclic $k$-generators with no directed cycle.\\

\noindent{\bf Chain reduction.} For $n\geq 3$, let $C=(x_1,x_2,\ldots,x_n)$ be a maximal $n$-chain that is common to $T$ and $T'$. Then set $S=T|X\setminus\{x_3,x_4,\ldots,x_n\}$ and $S'=T'|X\setminus\{x_3,x_4,\ldots,x_n\}$.\\

It is natural to ask whether or not the 3-2-chain reduction can also be applied when computing $r(T,T')$. The next lemma, whose proof  follows from the proof of Lemma~\ref{l:3-2-reduction} by considering maximum \fudge{\emph{acyclic}} agreement forests,
answers this question affirmatively. A slightly more general result was also established in~\cite{whidden2013fixed}.

\begin{lemma}\label{l:hybrid-3-2-reduction}
Let $T$ and $T'$ be two rooted phylogenetic $X$-trees, and let $S$ and $S'$ be two trees obtained from $T$ and $T'$, respectively, by a single application of the 3-2-chain reduction. Then $r(S,S')=r(T,T')-1$.
\end{lemma}

Let $T$ and $T'$ be two rooted phylogenetic $X$-trees, and let $S_1$ and $S'_1$ be two trees resulting from $T$ and $T'$, respectively, by \fudge{exhaustively} applying the subtree and chain reduction.
In~\cite{sempbordfpt2007}  a weight is associated to each 2-chain that results from applying the chain reduction. \delete{These weights are necessary to compute the size of a maximum acyclic agreement forest for $T$ and $T'$ given such a forest for $S_1$ and $S'_1$.} Hence, if we first apply a chain reduction and, subsequently, a 3-2-chain reduction, we \steven{would} need to take into account the weight of any previously reduced $n$-chain with $n\geq 3$. To avoid this, we establish the following.

\begin{lemma}\label{l:order}
Let $T$ and $T'$ be two rooted phylogenetic $X$-trees. Let $S_2$ and $S'_2$ be two trees obtained from $T$ and $T'$, respectively, by applying the subtree and 3-2-chain reduction until no such reduction is possible, and let $S$ and $S'$ be  two trees obtained from $S_2$ and $S'_2$, respectively, by applying the chain reduction until no further reduction is possible. Then none of the three reductions can be applied to $S$ and $S'$.
\end{lemma}

\begin{proof}
\fudge{
We make use of the following observations.
First, the chain reduction cannot create new common pendant subtrees.
Second, the chain reduction cannot
use
leaves from a weighted 2-chain created earlier \blue{since} this would contradict the maximality of the chain that was reduced earlier. Hence, we can view exhaustive applications of the chain reduction as simultaneously applying the reduction to a maximal set of leaf-disjoint maximal common chains in $S_2$ and $S'_2$, immediately yielding $S$ and $S'$. Clearly,}
$S$ and $S'$ do not have a common subtree or $n$-chain with $n\geq 3$. Assume that $S$ and $S'$ can be further reduced under the 3-2-chain reduction. Then there exist a pendant 3-chain $C_3=(x_1,x_2,x_3)$ in one of $S$ or $S'$, say $S$, and a pendant 2-chain $C_2=(x_3,x_i)$ with $i\in\{1,2\}$ in $S'$. Since $S_2$ and $S_2'$ cannot be reduced any further under the 3-2-chain reduction,  $C_3$ is not a pendant chain in $S_2$ or $C_2$ is not a pendant chain of $S_2'$;  \steven{and 
\blue{the existence of $C_1$ and $C_2$} is necessarily caused by leaves that are deleted by the chain reduction.}  
First, if $C_2$ is not pendant in $S_2'$, then there exists a pendant 2-chain $c\in\{(x_i,x_l),(x_3,x_l)\}$ in $S_2'$ \blue{with} \steven{$x_l\in X\setminus \{x_1,x_2,x_3\}$}, and  an $n$-chain $C$ with $n\geq 3$ such that $C$ is common to $S_2$ and $S_2'$ and the first two elements of $C$ are identical with those of $c$. In obtaining $S$ and $S'$ from $S_2$ and $S_2'$ respectively, $C$ is reduced to $c$; thereby contradicting that $C_2$ is pendant in $S'$.  Second, if $C_3$ is not pendant in $S_2$, then an element  $c\in\{(x_1,x_l),(x_2,x_l),(x_3,x_l),(x_1,x_2,x_l)\}$ is a pendant chain in $S_2$. Moreover, similar to the first case there exists an $n$-chain $C$ with $n\geq 3$ such that $C$ is common to $S_2$ and \steven{$S_2'$} and the first two (resp. three) elements of $C$ are identical with those in $c$. If $c\ne (x_1,x_2,x_l)$, then $C$ is reduced to a 2-chain that contains $x_l$; thereby contradicting that $C_3$ is pendant in $S'$. On the other hand, if $c= (x_1,x_2,x_l)$ then, as $C$ is common to  $S_2$ and $S_2'$, it follows that $(x_i,x_3,x_j)$ is a pendant 3-chain of $S_2'$ and $(x_i,x_j)$ is a pendant 2-chain of $S_2$, \steven{where $x_j$ is the leaf in $\{x_1, x_2\}$ not equal to $x_i$.} Hence $S_2$ and $S_2'$ can be
reduced by a 3-2-chain reduction;~a contradiction.\qed
\end{proof}

The next theorem can be established analogously to that of Theorem~\ref{t:kernel} by considering (i) $k$-generators and rooted phylogenetic networks \delete{instead of cyclic $k$-generators and rooted leaf-labeled graphs, respectively,} and (ii)
that at most \steven{\emph{two}} leaves can be attached \blue{to}  each edge side of a $k$-generator that is not directed into a vertex \blue{side}. 

\begin{theorem}
Let $S$ and $S'$ be two rooted phylogenetic $X$-trees such that $r(S,S')\geq 1$. 
Suppose that $S$ and $S'$ cannot be reduced any further by applying the subtree, 3-2-chain, or chain reduction. Then $|X|\leq 7r(S,S')-2$.
\end{theorem}

\fudge{This kernel} is again tight. \fudge{The proof is very similar to Theorem~\ref{t:tight}, which already uses an acyclic generator. We attach two leaves instead of three to edge sides and use
that $d_\rSPR(S_k,S_k')\leq r(S_k,S_k')$.}

\delete{Observe that neither the generator $G_k'$ that is shown in Fig.~\ref{fig:tight} nor the rooted leaf-labeled graph $G_k$ that is obtained from $G_k'$ by attaching leaves as described in proof of  Theorem~\ref{t:tight}  contains a directed cycle. Then following the same construction as that presented in the proof of Theorem~\ref{t:tight} but attaching two (instead of three) leaves to each edge side of $G_k'$ that is not directed into a vertex side results in a rooted phylogenetic network $G_k$ and  two rooted phylogenetic trees $S_k$ and $S_k'$ displayed by $G_k$ that each have size $7k-2$. The next theorem then follows by noting that $d_\rSPR(S_k,S_k')\leq r(S_k,S_k')$.}

\begin{theorem}
Let $S$ and $S'$ be two rooted phylogenetic $X$-trees such that $r(S,S')\geq 1$. 
Suppose that $S$ and $S'$ cannot be reduced any further by applying the subtree, 3-2-chain, or chain reduction. Then $|X|\leq 7r(S,S')-2$ is a tight bound.
\end{theorem}

\delete{\section{Final remarks}
We note that, if only the subtree and chain reductions are applied (but not the 3-2-chain reduction), our cyclic generator approach already gives a bound of $13k-3$. This improves upon the original bound of $O(28k)$ \cite{bordewich2005computational} which only used agreement forests, highlighting the extra combinatorial insight given by generators.} \delete{Finally, we remark that the analysis in this article is much simpler than the conceptually similar approach used for \emph{un}rooted trees in \cite{kelk2020new}. This is because, when obtaining a rooted phylogenetic tree from a (possibly cyclic) phylogenetic network, one of the two edges entering each reticulation must be deleted. Since unrooted trees are undirected, tree vertices and reticulations are indistinguishable  in the underlying network/generator. The design of reduction rules and their analysis must take this greater freedom into account.}

\end{document}